\documentclass[11pt]{article}

\usepackage[margin=1in]{geometry}
\usepackage{amsmath,amssymb,amsthm}
\usepackage{mathtools}
\usepackage{booktabs}
\usepackage{enumitem}
\usepackage[colorlinks=true,linkcolor=blue,citecolor=blue,urlcolor=blue]{hyperref}

\theoremstyle{plain}
\newtheorem{theorem}{Theorem}
\newtheorem{lemma}{Lemma}
\newtheorem{proposition}{Proposition}

\theoremstyle{definition}
\newtheorem{definition}{Definition}
\newtheorem{remark}{Remark}

\DeclareMathOperator{\fix}{fix}
\DeclareMathOperator{\cfix}{cfix}
\DeclareMathOperator{\Aut}{Aut}
\DeclareMathOperator{\Stab}{Stab}
\DeclareMathOperator{\lcm}{lcm}
\newcommand{\Z}{\mathbb{Z}}
\newcommand{\orb}{\mathcal{B}}
\newcommand{\gtil}{\tilde g}
\newcommand{\ghat}{\hat g}
\newcommand{\ftil}{\tilde f}
\newcommand{\fhat}{\hat f}

\title{\bf Counting Unlabeled Chordal Graphs by Equivariant Evaporation}
\author{Matthew Sun\\[2pt]\normalsize Independent researcher, California, USA\\\normalsize\texttt{matthewsun42@gmail.com} / \texttt{mattsun1@mit.edu}}
\date{June 2026}

\begin{document}
\maketitle

\begin{abstract}
We compute the number of unlabeled chordal graphs on $n$ vertices, both the total count
(\href{https://oeis.org/A048193}{OEIS A048193}) and the connected count
(\href{https://oeis.org/A048192}{OEIS A048192}), extending two sequences whose published values had
remained at $n=15$. The method is a P\'olya--Burnside enumeration: the number of unlabeled graphs in
a class closed under relabeling is the average over $S_n$ of the number of labeled graphs fixed by
each permutation. The technical core is the evaluation, for an arbitrary permutation $\pi$, of the
number of $\pi$-invariant labeled chordal graphs. We give a dynamic program for this quantity that
lifts the evaporation-based labeled chordal counting of H\'ebert-Johnson, Lokshtanov and Vigoda to
the equivariant setting. Its central structural ingredient is a \emph{divisor-bundle decomposition}:
when a connected piece spans a cyclic orbit of size $c$, it forms, for each divisor $d\mid c$, a
$d$-fold bundle whose constituent is an object of the same kind in the cyclic world of order $c/d$,
computed by the same program recursively. We prove the decomposition and the correctness of the
resulting recurrences, and we prove that the full Burnside computation runs in sub-exponential time
$n^{O(\sqrt n)}$. We report the new terms through $n=20$ and describe four independent validations,
including exact agreement with all previously known values of both sequences and an Euler-transform
consistency check.
\end{abstract}

\section{Introduction}

A graph is \emph{chordal} if it has no induced cycle of length at least four. Chordal graphs are a
central class in algorithmic graph theory, characterized by perfect elimination orderings and by
tree decompositions whose width equals the clique number minus one. Two integer sequences count them
up to isomorphism:
\begin{align*}
\text{\href{https://oeis.org/A048193}{A048193}}(n)&=\#\{\text{chordal graphs on }n\text{ nodes}\},\\
\text{\href{https://oeis.org/A048192}{A048192}}(n)&=\#\{\text{connected chordal graphs on }n\text{ nodes}\},
\end{align*}
both up to isomorphism. Their published data extended only to $n=15$.

For the \emph{labeled} count, H\'ebert-Johnson, Lokshtanov and Vigoda~\cite{HLV2023} gave a
polynomial-time dynamic program based on an \emph{evaporation} process that repeatedly deletes all
simplicial vertices. Counting graphs \emph{up to isomorphism}, however, is governed by Burnside's
lemma and requires counting labeled chordal graphs invariant under each permutation of the vertices.
H\'ebert-Johnson and Lokshtanov~\cite{HLV2025} proved that counting labeled chordal graphs with a
prescribed automorphism is fixed-parameter tractable in the number of moved points $\mu$, with a
running time of $O(2^{7\mu}n^{9})$; in the regime relevant to the Burnside sum, where $\mu$ may be
as large as $n$, that bound is single-exponential in $n$, and no exact unlabeled enumeration was
carried out there.

\paragraph{Contributions.} Building on the labeled evaporation dynamic program of~\cite{HLV2023},
we give a concrete equivariant dynamic program for the number $\fix(\pi)$ of $\pi$-invariant labeled
chordal graphs, and use it to compute A048192 and A048193 beyond the previous frontier. Our results
are:
\begin{enumerate}[topsep=2pt,itemsep=1pt]
\item A \emph{divisor-bundle} structure theorem (the Component-Orbit Decomposition,
Theorem~\ref{thm:cod}) describing how connected pieces of a $\pi$-invariant graph decompose, and
reducing each equivariant recurrence to its labeled counterpart of~\cite{HLV2023}
(Section~\ref{sec:lift}).
\item A correctness theorem (Theorem~\ref{thm:correct}) for the assembled program, conditional only
on the labeled recurrences of~\cite{HLV2023}.
\item A running-time theorem (Theorem~\ref{thm:time}): the entire computation of A048192$(n)$ and
A048193$(n)$ runs in sub-exponential time $n^{O(\sqrt n)}=e^{O(\sqrt n\log n)}$, in contrast to the
single-exponential worst case of the general parameterized bound.
\item The new values through $n=20$ (Section~\ref{sec:results}), with four independent validations.
\end{enumerate}

\section{Preliminaries}

\subsection{The Burnside reduction}

Let $\mathcal H$ be a class of graphs on a fixed vertex set that is closed under relabeling
(chordality is such a property). The number of isomorphism classes of $\mathcal H$-graphs on $n$
vertices is, by the Cauchy--Frobenius--Burnside lemma,
\begin{equation}\label{eq:burnside}
a(n)=\frac{1}{n!}\sum_{\pi\in S_n}\fix(\pi)
   =\sum_{\lambda\vdash n}\frac{\fix(\lambda)}{z_\lambda},
\qquad z_\lambda=\prod_{i\ge 1} i^{m_i}\,m_i!,
\end{equation}
where $\fix(\pi)$ is the number of $\mathcal H$-graphs invariant under $\pi$ (i.e.\ with
$\pi\in\Aut(G)$), a quantity that depends only on the cycle type $\lambda=(1^{m_1}2^{m_2}\cdots)$ of
$\pi$, and $z_\lambda=n!/|\{\pi:\text{type }\lambda\}|$. The connected count obeys the same identity
with $\fix$ replaced by $\cfix$, the number of \emph{connected} invariant graphs, since $S_n$ acts on
connected labeled graphs as well. Thus everything reduces to computing $\fix(\lambda)$ and
$\cfix(\lambda)$ for each cycle type $\lambda$.

Throughout, fix $\pi\in S_V$ with $V$ the vertex set, $|V|=n$, and let $N=\lcm$ of the cycle lengths
$=\operatorname{ord}(\pi)$. The orbits of $\pi$ partition $V$; an orbit of size $c$ is a
\emph{$c$-cycle}, on which $\pi$ acts as a single cycle. A part of size $c$ in $\lambda$ is such an
orbit.

\subsection{Evaporation}

\begin{definition}[Evaporation with an exception clique; \cite{HLV2023}]\label{def:evap}
A vertex $v$ of a graph $G$ is \emph{simplicial} if its neighborhood $N_G(v)$ is a clique. Given a
clique $X\subseteq V(G)$, set $G_0=G$ and, for $t\ge1$, let $L_t$ be the set of simplicial vertices
of $G_{t-1}$ not in $X$ and put $G_t=G_{t-1}\setminus L_t$. For chordal $G$ this terminates with
$V(G_t)=X$; the sequence $L_1,\dots,L_T$ is the \emph{evaporation of $G$ with exception $X$}, $G$
\emph{evaporates at time $T$}, the \emph{evaporation time} of $v\notin X$ is the index $t$ with
$v\in L_t$, and $L_G(X):=L_T$ is the last layer.
\end{definition}

We write $\tau(v)$ for the evaporation time of $v$.

\section{The orbit-level process and its structure}

We run the evaporation process on a $\pi$-invariant chordal graph with a $\pi$-invariant exception
clique. The following lemma is what makes an equivariant dynamic program well-defined: the process
acts on $\pi$-orbits, not on individual vertices.

\begin{lemma}[Orbit-evaporation]\label{lem:orbit}
Let $G$ be chordal with $\pi\in\Aut(G)$, and let the exception clique $X$ be $\pi$-invariant. Then
$\pi(L_t)=L_t$ and $\pi\in\Aut(G_t)$ for every $t$. Consequently $\tau$ is constant on each
$\pi$-orbit: every orbit evaporates in a single layer.
\end{lemma}

\begin{proof}
For any $\pi$-invariant graph $H$ one has $N_H(\pi v)=\pi(N_H(v))$, and $\pi$ maps cliques to
cliques; hence $v$ is simplicial in $H$ iff $\pi v$ is. We induct on $t$. Base: $G_0=G$ is
$\pi$-invariant. Step: if $\pi\in\Aut(G_{t-1})$ then simpliciality in $G_{t-1}$ is $\pi$-stable and
$X$ is $\pi$-invariant, so $\pi(L_t)=L_t$; deleting the $\pi$-invariant set $L_t$ leaves $\pi$ an
adjacency-preserving bijection of $V(G_t)$, i.e.\ $\pi\in\Aut(G_t)$. Since $\pi(L_t)=L_t$, each
$\pi$-orbit is contained in or disjoint from $L_t$, so $\tau$ is constant on orbits.
\end{proof}

A connected component need not be $\pi$-invariant; rather $\pi$ permutes the components. The orbit of
a component is a \emph{bundle}.

\begin{lemma}[Component-orbit bundle]\label{lem:bundle}
Let $G$ be $\pi$-invariant and $K$ a connected component of $G$, with $\pi$-orbit
$\orb=\{K,\pi K,\dots,\pi^{d-1}K\}$ of size $d$. Then:
\begin{enumerate}[topsep=1pt,itemsep=0pt,label=\textup{(\roman*)}]
\item $d\mid N$ and $\Stab(K)=\langle\pi^{d}\rangle$, so $K$ is $\pi^{d}$-invariant;
\item the $d$ components of $\orb$ are pairwise isomorphic and $\bigcup_j\pi^jK$ is $\pi$-invariant.
\end{enumerate}
\end{lemma}

\begin{proof}
$\langle\pi\rangle$ is cyclic of order $N$. By orbit--stabilizer $[\langle\pi\rangle:\Stab(K)]=d$,
so $|\Stab(K)|=N/d$ and $d\mid N$. A cyclic group has a unique subgroup of each order dividing $N$;
the one of order $N/d$ is $\langle\pi^d\rangle$, giving (i). Each $\pi^j$ restricts to an isomorphism
$K\to\pi^jK$, and a union of a $\langle\pi\rangle$-orbit is $\pi$-invariant, giving (ii).
\end{proof}

\begin{lemma}[Block decomposition]\label{lem:block}
Let $O$ be a $\pi$-orbit of size $s$ and $H=\langle\pi^d\rangle$ for $d\mid N$. Then $H$ acts on $O$
with exactly $\gcd(d,s)$ orbits (\emph{blocks}), each of size $s/\gcd(d,s)$. Any
$\pi^{d}$-invariant set meets $O$ in a union of whole blocks.
\end{lemma}

\begin{proof}
Identify $O\cong\Z/s\Z$ with $\pi|_O$ translation by $1$, so $\pi^{d}|_O$ is translation by $d$. The
orbits of $\langle d\rangle\le\Z/s\Z$ are the cosets of the subgroup generated by $d$, which has
order $s/\gcd(d,s)$; there are $\gcd(d,s)$ of them. A $\pi^d$-invariant set is a union of such
orbits.
\end{proof}

\begin{lemma}[Recursive sub-world]\label{lem:sub}
With $K$, $\orb$, $d$ as in Lemma~\ref{lem:bundle}, the component $K$ is a connected chordal graph
invariant under the cyclic group $H=\langle\pi^{d}\rangle$ of order $N/d$. Every $\pi$-orbit $O$ of
size $s$ that $K$ meets contributes to $K$ a union of blocks of size $s/\gcd(d,s)$. Moreover, if $O$
has size $c$ with $d\mid c$ and is one of the orbits spanned by $\orb$, then $K$ meets $O$ in exactly
one block, a single $H$-orbit of size $c/d$. Hence $K$ is an object of the same type on a strictly
smaller cycle type, and the bundle structure recurses, terminating at $d=N$ in the labeled world.
\end{lemma}

\begin{proof}
$K$ is $H$-invariant (Lemma~\ref{lem:bundle}(i)) and $H$ is cyclic, so by Lemma~\ref{lem:block} $K\cap
O$ is a union of blocks of size $s/\gcd(d,s)$. For a spanned $O$ with $d\mid c$ we have
$\gcd(d,c)=d$, so $O$ has exactly $d$ blocks, each of size $c/d$. The bundle has period $d$, so its
$d$ components are distinct and partition $\bigcup_j\pi^jK\supseteq O$. As $K\cap O\ne\varnothing$ and
$\pi O=O$, each translate meets $O$ in $\pi^j(K\cap O)\ne\varnothing$; these $d$ sets are pairwise
disjoint and cover $O$, and $\pi^j$ is a bijection, so $d\,|K\cap O|=|O|=c$ and $|K\cap O|=c/d$. Since
$K\cap O$ is a union of size-$(c/d)$ blocks of total size $c/d$, it is exactly one block. The spanned
orbit of size $c$ thus becomes a sub-orbit of size $c/d<c$ (as $d>1$), so the recursion is
well-founded and ends when $d=N$, where $H=\{\mathrm{id}\}$ and $K$ is an unconstrained labeled graph.
\end{proof}

\section{The Component-Orbit Decomposition}\label{sec:cod}

For a cyclic world of order $M$ and a cycle type $\nu$, let $g_1^{\langle M\rangle}(t,x,\nu)$ denote
the number of configurations consisting of a single connected component-bundle on the new orbits
$\nu$, seeing a prescribed boundary, with evaporation time $t$ (Section~\ref{sec:dp} makes the
boundary parameters precise; for the present theorem only its existence is used). The labeled count
$\sum_t g_1^{\langle 1\rangle}(t,0,1^k)=\text{\href{https://oeis.org/A007134}{A007134}}(k)$ is the
number of connected labeled chordal graphs on $k$ vertices~\cite{HLV2023}.

\begin{theorem}[Component-Orbit Decomposition]\label{thm:cod}
Let $\mu$ be a multiset of new orbits and $d$ a positive integer dividing every size occurring in
$\mu$. The number of period-$d$ component-bundles spanning exactly $\mu$, seeing a prescribed set of
boundary blocks, with evaporation time $t$, equals
\[
d^{\,|\mu|-1}\; g_1^{\langle N/d\rangle}\!\big(t,\;x',\;\mu/d\big),
\]
where $\mu/d$ replaces each size $s$ by $s/d$, the order-$(N/d)$ world is that of
Lemma~\ref{lem:sub}, and the boundary $x'$ records, for each boundary orbit of size $s$, its
$\gcd(d,s)$ blocks of sub-type $s/\gcd(d,s)$. For $d=1$ the term is the equivariant count $g_1(t,x,\mu)$.
\end{theorem}

\begin{proof}
By Lemma~\ref{lem:sub} a representative $K$ of such a bundle is, with its incidences to the boundary
blocks, a connected configuration in the order-$(N/d)$ world on the sub-orbits $\mu/d$ (each spanned
size-$s$ orbit met in a single sub-orbit of size $s/d$) and on the boundary blocks $x'$; its
evaporation time is $t$ because, by Lemma~\ref{lem:orbit} applied to $\langle\pi\rangle\supseteq
\langle\pi^d\rangle$, $\tau$ is constant on the $\langle\pi^d\rangle$-orbits that are the sub-vertices
of $K$, so the sub-world evaporation has the same layers. Conversely the map $K\mapsto\bigcup_j\pi^jK$
sends each such sub-configuration to a period-$d$ bundle (a disjoint union of the chordal $\pi^jK$,
$\pi$-invariant by Lemma~\ref{lem:bundle}(ii)).

It remains to count the bundles. Fix a distinguished spanned orbit $O_1$ and let $w_1$ be its least
vertex. We exhibit a bijection
\[
\big\{\text{period-}d\text{ bundles on }\mu\big\}
\ \longleftrightarrow\
\big\{\text{order-}(N/d)\text{ sub-configurations}\big\}\times\{1,\dots,d\}^{\,|\mu|-1}.
\]

\emph{Forward map.} Given a bundle $\orb$, let $K$ be the component containing $w_1$ (well-defined:
$w_1$ lies in exactly one of the $d$ blocks of $O_1$, hence in exactly one component, since the
translates partition $O_1$ by the argument of Lemma~\ref{lem:sub}). Record the isomorphism type of
$K$ as an order-$(N/d)$ sub-configuration, and the alignment $a=(a_i)_{i\ge2}$, where $a_i\in\{1,\dots,d\}$
names which block of $O_i$ lies in $K$ (each $O_i$ contributes exactly one block to $K$ by
Lemma~\ref{lem:sub}); the block of $O_1$ in $K$ is fixed to be the one containing $w_1$.

\emph{The period is exactly $d$.} For $0<j<d$, $\pi^{j}K$ meets $O_1$ in $\pi^{j}(K\cap O_1)$, a
block distinct from $K\cap O_1$ (the $d$ blocks of $O_1$ are permuted freely by $\pi$, as in
Lemma~\ref{lem:sub}); hence $\pi^{j}K\neq K$, so $\orb$ has exactly $d$ distinct components and the
forward map lands in period-$d$ bundles regardless of any abstract automorphism of $K$. This rules
out the only possible source of over-collapse — a representative with internal periodic symmetry —
because the symmetry that would matter, $\pi^{j}K=K$, is forbidden by the block displacement.

\emph{Injectivity.} The type of $K$ and the alignment determine the vertex set and edges of $K$
exactly (the alignment fixes which concrete blocks of each $O_i$ are in $K$, and the type fixes the
adjacencies among them and to the boundary blocks); $K$ then determines $\orb=\{K,\pi K,\dots,\pi^{d-1}K\}$.

\emph{Surjectivity and freeness.} Conversely, any sub-configuration $K$ and alignment
$a\in\{1,\dots,d\}^{|\mu|-1}$ yields a bundle: place $K$ on the blocks named by $a$ (and the
$w_1$-block of $O_1$), set $\orb=\bigcup_j\pi^jK$. This is $\pi$-invariant (Lemma~\ref{lem:bundle}),
chordal (disjoint union of chordal pieces), of period $d$ (above), and spans $\mu$ with the
prescribed block-seeing. The alignment choices are unconstrained: changing $a_i$ relabels the
order-$(N/d)$ sub-orbit of $K$ on $O_i$ to a different block of $O_i$, a graph isomorphism that
preserves connectivity and all boundary incidences (those are recorded against boundary
\emph{blocks}, which are fixed and shared by all components, hence unaffected by the new-orbit
alignment). Thus the boundary-seeing data and the alignment are independent, and there are exactly
$d^{\,|\mu|-1}$ alignments.

Therefore the number of bundles is $d^{\,|\mu|-1}$ times the number of order-$(N/d)$
sub-configurations of the prescribed type, which is $g_1^{\langle N/d\rangle}(t,x',\mu/d)$.
\end{proof}

Summing over admissible periods gives the count of all bundles spanning $\mu$:
\begin{equation}\label{eq:codsum}
\text{COD}(t,x,\mu)=\underbrace{g_1(t,x,\mu)}_{d=1}
   +\sum_{\substack{d>1\\ d\mid\gcd(\mu)}} d^{\,|\mu|-1}\,g_1^{\langle N/d\rangle}(t,x',\mu/d),
\end{equation}
where the boundary parameter $x'$ and the admissible boundary set are specialized to the condition
in force (Section~\ref{sec:lift}).

\section{The dynamic program and its correctness}\label{sec:dp}

\subsection{State variables and the block decomposition}\label{sec:state}

Fix the distinct orbit sizes $S=(s_1<\cdots<s_r)$ of $\pi$. All state variables are integer vectors
indexed by $S$:
\[
x,\,k,\,z,\,l\ \in\ \Z_{\ge0}^{\,r},\qquad z\le x\ \text{(componentwise)}.
\]
Here $x_i$ is the number of \emph{boundary} orbits of size $s_i$ (their union $X$ is a clique);
$k_i$ is the number of \emph{new} orbits of size $s_i$; $l_i$ (for the $f$-functions) is the number
of orbits in the last evaporation layer $L_G(X)$; and $z$ is an inclusion--exclusion sub-boundary,
specifying — under a fixed linear order on the boundary orbits — a downward-closed sub-clique used to
express ``the piece reaches beyond $z$.'' A vertex count is recovered as $|X|=\sum_i s_i x_i$, etc.
There is no separate ``seeing'' state variable: which boundary a piece sees is summed over inside the
recurrences (below), exactly as in the labeled case, where it is the variable $x'\le x$.

The one genuinely new ingredient is how a \emph{bundle} interacts with the boundary. By
Lemma~\ref{lem:block}, a period-$d$ bundle does not see boundary \emph{orbits} but boundary
\emph{blocks}: a boundary orbit of size $s_i$ presents $\gcd(d,s_i)$ blocks of sub-type
$\tau_i:=s_i/\gcd(d,s_i)$. We therefore define, for each $d\mid N$, the \emph{block profile} of the
boundary $x$ as the vector indexed by the distinct sub-types $\{\tau_i\}$:
\begin{equation}\label{eq:blockprofile}
\beta_d(x)_\tau \;=\; \sum_{i:\ \tau_i=\tau}\gcd(d,s_i)\,x_i
\qquad(\text{number of blocks of sub-type }\tau),
\end{equation}
and the associated boundary generating function in formal variables $\{v_\tau\}$,
\begin{equation}\label{eq:gf}
\Gamma_d(x)\;=\;\prod_{i=1}^{r}\big(1+v_{\tau_i}\big)^{\gcd(d,s_i)\,x_i},
\qquad
\Gamma_d^{\mathrm{cov}}(x)\;=\;\prod_{i=1}^{r}\Big(\big(1+v_{\tau_i}\big)^{\gcd(d,s_i)}-1\Big)^{x_i}.
\end{equation}
The coefficient of $\prod_\tau v_\tau^{\,a_\tau}$ in $\Gamma_d(x)$ is the number of ways a bundle's
sub-component can see exactly $a_\tau$ blocks of each sub-type; $\Gamma_d^{\mathrm{cov}}$ is the same
restricted to seeing at least one block of every boundary orbit (``covers all of $X$''). The
parenthetical is exact, not merely suggestive: by Lemma~\ref{lem:block} the quotient
$\langle\pi\rangle/\langle\pi^{d}\rangle\cong\Z/\gcd(d,s_i)\Z$ permutes the $\gcd(d,s_i)$ blocks of
each boundary orbit $O_i$ transitively, so the full period-$d$ bundle
$K\cup\pi K\cup\dots\cup\pi^{d-1}K$ is adjacent to all of $O_i$ exactly when its representative $K$
meets at least one block of $O_i$ (the $d$ translates sweep that block across the whole orbit).
Thus $\Gamma_d^{\mathrm{cov}}$ counts component-orbits that, taken as a whole, see all of $X$, even
when no single constituent component does. These are
the equivariant analogues of the labeled coefficients $\binom{x}{x'}$ and $\binom{x}{x}$, and the
sub-type vector $x'=(a_\tau)_\tau$ is the boundary argument of the sub-world function in
Theorem~\ref{thm:cod}. With $d=1$ every block is a whole orbit ($\tau_i=s_i$, $\gcd=1$) and
$\Gamma_1(x)=\prod_i(1+v_{s_i})^{x_i}$ reduces to choosing a sub-vector $x'\le x$ of boundary orbits,
recovering the labeled situation.

\subsection{The counter functions}

We use the eight functions of~\cite[Def.\ 3.3]{HLV2023}, read on the orbit-level evaporation of
Lemma~\ref{lem:orbit} over the state variables above; ``component'' means component-bundle. Briefly,
for a $\pi$-invariant chordal $G$ with $X$ a clique:
\begin{center}\small
\begin{tabular}{ll}
\toprule
$g(t,x,k,z)$ & every component of $G\setminus X$ sees $X\setminus[z]$; evaporates in time $\le t$\\
$\gtil(t,x,k,z)$ & as $g$, with every component evaporating at time exactly $t$\\
$\ghat(t,x,k,z)$ & as $\gtil$, with no component seeing all of $X$\\
$g_1(t,x,k)$ & as $\gtil$, exactly one component, seeing all of $X$\\
$g_2(t,x,k)$ & as $g_1$ but $\ge 2$ components, each seeing all of $X$\\
$f(t,x,l,k)$ & connected; evaporates at exactly $t$; $L_G(X)=[x{+}1,x{+}l]$; $X\cup L_G(X)$ a clique\\
$\ftil(t,x,l,k)$ & as $f$, components of $G\setminus(X\cup L)$ evaporate at exactly $t{-}1$ ($\ge1$)\\
$\fhat(t,x,l,k)$ & as $\ftil$, with no component of $G\setminus(X\cup L)$ seeing all of $X\cup L$\\
\bottomrule
\end{tabular}
\end{center}
The labeled instance ($\pi=\mathrm{id}$) is exactly~\cite{HLV2023}; there
$\cfix_{\mathrm{lab}}(k)=\sum_t g_1(t,0,k)$.

\begin{theorem}[Labeled recurrences; \cite{HLV2023}, Lemmas 3.4--3.8]\label{thm:hlv}
For $\pi=\mathrm{id}$ the eight functions satisfy the recurrences
\begin{align*}
g_1(t,x,k)&=\textstyle\sum_{l=1}^{k}\binom{k}{l}f(t,x,l,k{-}l),\\
f(t,x,l,k)&=\textstyle\sum_{k'=1}^{k}\binom{k}{k'}\ftil(t,x,l,k')\,g(t{-}2,x{+}l,k{-}k',x),\\
g(t,x,k,z)&=\textstyle\sum_{k'=0}^{k}\binom{k}{k'}\gtil(t,x,k',z)\,g(t{-}1,x,k{-}k',z),\\
\gtil(t,x,k,z)&=\textstyle\sum_{k'=1}^{k}\sum_{x'=1}^{x}\!\big(\binom{x}{x'}{-}\binom{z}{x'}\big)\binom{k-1}{k'-1}g_1(t,x',k')\,\gtil(t,x,k{-}k',z),
\end{align*}
together with the stated recurrences for $\ghat$, $g_2$, and $\ftil$ (Lemma 3.8), all proved by
conditioning on the component containing the least non-boundary vertex and recursing on the rest.
\end{theorem}

\subsection{From labeled to equivariant}\label{sec:lift}

We lift Theorem~\ref{thm:hlv} to general $\pi$; the resulting eight equivariant recurrences are
written explicitly in Appendix~\ref{app:rec}. The four recurrences that compose other functions
($g_1,f,g,\ftil$) lift verbatim: their derivations partition the new vertices by structural role
(membership in $L_G(X)$, evaporation time, or seeing-all status), and by Lemma~\ref{lem:orbit} these
roles are orbit-stable, so the per-size product of binomials
$\binom{k}{\cdot}=\prod_i\binom{k_i}{\cdot}$ replaces the single binomial and each constituent
function is its equivariant counterpart. The remaining recurrences condition on a single component;
for them we use the following principle.

\begin{proposition}[Lift Principle]\label{prop:lift}
Let $F\in\{\gtil,\ghat,g_2,\fhat\}$, defined by a condition $\Phi_F$ on the components of $G\setminus
X$ (an evaporation pattern and a boundary-seeing pattern). The equivariant recurrence for $F$ is the
recurrence of Theorem~\ref{thm:hlv} for $F$ with the single-component count $g_1(t,x',k')$ replaced by
$\mathrm{COD}_{\Phi_F}(t,x',\kappa)$ of \eqref{eq:codsum}, restricted to the boundary condition of
$\Phi_F$.
\end{proposition}

\begin{proof}
Fix the distinguished orbit $o^\ast$ of largest present size (least vertex among ties). By
Lemmas~\ref{lem:orbit}--\ref{lem:bundle} the components of $G\setminus X$ form component-bundles and
$o^\ast$ lies in a unique one, $\mathcal C^\ast$, spanning some $\kappa$ with $\kappa_{p}\ge1$ at the
marked size $p$. Conditioning on $\mathcal C^\ast$: the orbits joining it are counted by
$\binom{k_p-1}{\kappa_p-1}\prod_{i\ne p}\binom{k_i}{\kappa_i}$ (the marked orbit forced in, the rest
free); $\mathcal C^\ast$ itself is counted by $\mathrm{COD}_{\Phi_F}(t,x',\kappa)$ by
Theorem~\ref{thm:cod}, with the boundary coefficient of $\Phi_F$ applied per orbit for the $d=1$ term
and per block for $d>1$ (Lemma~\ref{lem:block}); and $G\setminus(X\cup\mathcal C^\ast)$ is an
$F$-configuration on the remaining orbits, counted by the recursive $F$-term. The three are
independent (disjoint vertex sets) and the counts multiply, exactly as in the proof of
Theorem~\ref{thm:hlv}; the only change is the substitution of $\mathrm{COD}$ for the single-component
count.
\end{proof}

\begin{proposition}[All from connected]\label{prop:fix}
Write the cycle type as $k=(k_1,\dots,k_r)$ over the distinct sizes $s_1<\cdots<s_r$, let
$p=\max\{i:k_i>0\}$, and let $\mathrm{BUNDLE}(\mu)$ be the number of single component-bundles
spanning $\mu$, namely
\[
\mathrm{BUNDLE}(\mu)=\cfix(\mu)+\sum_{d>1,\,d\mid\gcd(\mu)}d^{\,|\mu|-1}\,\cfix^{\langle N/d\rangle}(\mu/d).
\]
Then
\[
\fix(k)=\sum_{\substack{\mu\le k\\ \mu_p\ge1}}\Big[\tbinom{k_p-1}{\mu_p-1}\textstyle\prod_{i\ne p}\binom{k_i}{\mu_i}\Big]\,
   \mathrm{BUNDLE}(\mu)\,\fix(k-\mu),\qquad \fix(\mathbf 0)=1.
\]
\end{proposition}

\begin{proof}
A $\pi$-invariant chordal graph is the disjoint union of its component-bundles
(Lemma~\ref{lem:bundle}), and chordality is tested componentwise. Distinguish the orbit $o^\ast$ of
size $s_p$ with least vertex; it lies in a unique bundle, spanning $\mu$ with $\mu_p\ge1$. Among the
$k_p$ orbits of size $s_p$, $o^\ast$ is forced in and the other $\mu_p-1$ chosen from $k_p-1$; for
$i\ne p$ the $\mu_i$ orbits are any of the $k_i$. Given the orbit set, the bundle ranges over
$\mathrm{BUNDLE}(\mu)$ possibilities and the remainder independently over $\fix(k-\mu)$. Summing over
$\mu$ gives the identity.
\end{proof}

\begin{theorem}[Correctness]\label{thm:correct}
Assume Theorem~\ref{thm:hlv}. For every cycle type $\lambda$ with orbit vector $k$,
\[
\cfix(\lambda)=\sum_t g_1(t,0,k),\qquad
\fix(\lambda)=\text{Proposition~\ref{prop:fix}},
\]
and hence $\textup{A048192}(n)=\sum_{\lambda\vdash n}\cfix(\lambda)/z_\lambda$ and
$\textup{A048193}(n)=\sum_{\lambda\vdash n}\fix(\lambda)/z_\lambda$ by~\eqref{eq:burnside}.
\end{theorem}

\begin{proof}
By Lemma~\ref{lem:orbit} the orbit-level evaporation and all eight conditions are well-defined. The
four composition recurrences lift verbatim and the four single-component recurrences hold by the
Lift Principle (Proposition~\ref{prop:lift}) with the Component-Orbit Decomposition
(Theorem~\ref{thm:cod}); the recursion on the order-$(N/d)$ sub-worlds is well-founded
(Lemma~\ref{lem:sub}) with base case $\pi=\mathrm{id}$, where the functions and recurrences are those
of Theorem~\ref{thm:hlv}. Therefore $\sum_t g_1(t,0,k)$ counts the connected $\pi$-invariant chordal
graphs and Proposition~\ref{prop:fix} counts all of them.
\end{proof}

\section{Running time}\label{sec:time}

For a cycle type $\lambda$ with distinct sizes $s_1<\cdots<s_r$ and multiplicities $m_i$, write
$P(\lambda)=\prod_{i=1}^r(m_i+1)$.

\begin{lemma}\label{lem:percost}
Let $r=r(\lambda)$ be the number of distinct part sizes of $\lambda$. The dynamic program for
$\lambda$, including its divisor-bundle sub-instances, runs in time
\[
T(\lambda)\ \le\ n^{O(1)}\,P(\lambda)^{O(1)}\,(n+1)^{\,O(r)}.
\]
\end{lemma}

\begin{proof}
\emph{States.} Each function is memoized over an index $t\le n$ and between one and four tuples
($x,k,z$, and $l$ for the $f$-functions), each bounded coordinatewise by the multiplicities, so
taking at most $P(\lambda)$ values; there are $O\!\big(n\,P(\lambda)^4\big)$ states per function.

\emph{Block generating functions.} A recurrence's bundle term builds a multivariate generating
function~\eqref{eq:gf} in $r'\le r$ block variables of total degree at most the number of blocks
$\sum_i\gcd(d,s_i)x_i\le n$. Such a polynomial has at most $\prod_\tau(\beta_d(x)_\tau+1)\le(n+1)^{r}$
monomials, and is assembled in $(n+1)^{O(r)}$ ring operations. (This corrects a naive
$n^{O(1)}$ estimate: the monomial count is exponential in the number of distinct block sub-types.)

\emph{Work per state.} Besides the generating function, a recurrence sums over a sub-multiset of $k$
($\le P(\lambda)$ terms; $\le P(\lambda)^2$ for the absorption recurrence) and over divisors
(at most $d(c)=n^{o(1)}\le n$), and evaluates the sub-world $g_1$ at each generating-function monomial (a memo
lookup). All counts are non-negative integers bounded by the number of labeled chordal graphs on $n$
vertices, $\le 2^{\binom n2}$, hence of $O(n^2)$ bits, so each arithmetic operation costs
$\mathrm{poly}(n)$. Work per state is therefore $n^{O(1)}\,P(\lambda)^{2}\,(n+1)^{r}$.

\emph{Sub-instances.} A bundle term calls a sub-instance whose multiplicity vector has
$P(\lambda')\le P(\lambda)$ (the map $s\mapsto s/\gcd(d,s)$ only merges parts, and
$(a{+}1)(b{+}1)\ge a{+}b{+}1$) and at most $r$ distinct sizes; the recursion depth is $\le\log_2 n$
(each step at least halves the largest size). The number of distinct reachable sub-instances is at
most $\prod_i d(s_i)\le 2^{O(r)}$; charging each by the same bound multiplies the total by $2^{O(r)}$,
which is absorbed into the $(n+1)^{O(r)}$ factor (as $(n+1)^{O(r)}\ge 2^{O(r)}$), leaving the stated
form.
\end{proof}

\begin{lemma}\label{lem:maxprod}
$\displaystyle\max_{\lambda\vdash n}P(\lambda)\le e^{(\sqrt2+o(1))\sqrt n}=e^{O(\sqrt n)}$, and the
number of distinct part sizes satisfies $r(\lambda)\le\sqrt{2n}$.
\end{lemma}

\begin{proof}
The distinct sizes satisfy $s_i\ge i$, so $\sum_i i\,m_i\le\sum_i s_i m_i=n$. Relaxing $m_i\ge0$ to
the reals can only increase $\max\sum_i\log(1+m_i)$ subject to $\sum_i i\,m_i\le n$. The relaxed
objective is concave; a Lagrange multiplier gives $1+m_i=1/(\mu i)$ for $i\le C:=1/\mu$ and $m_i=0$
beyond, and the constraint $\sum_{i\le C}(\tfrac1\mu-i)\le n$ yields $C^2/2\,(1+o(1))\le n$, i.e.\
$C\le(\sqrt2+o(1))\sqrt n$. Then $\sum_{i\le C}\log\frac{C}{i}=C\log C-\log(C!)=C+O(\log C)$ by
Stirling, so $\log\max_\lambda P(\lambda)\le(\sqrt2+o(1))\sqrt n$. Finally, since the $s_i$ are
distinct, $1+2+\cdots+r\le\sum_i s_i\le\sum_i s_i m_i=n$, whence $r(r+1)/2\le n$ and $r\le\sqrt{2n}$.
\end{proof}

\begin{theorem}[Sub-exponential running time]\label{thm:time}
Computing $\textup{A048192}(n)$ and $\textup{A048193}(n)$ by the equivariant Burnside method takes
time $n^{O(\sqrt n)}=e^{O(\sqrt n\log n)}$.
\end{theorem}

\begin{proof}
The total time is $\sum_{\lambda\vdash n}T(\lambda)$ plus polynomial Burnside bookkeeping. By
Lemma~\ref{lem:percost}, $T(\lambda)\le n^{O(1)}P(\lambda)^{O(1)}(n+1)^{O(r)}$. By
Lemma~\ref{lem:maxprod}, $P(\lambda)^{O(1)}\le e^{O(\sqrt n)}$ and $r\le\sqrt{2n}$, so
$(n+1)^{O(r)}\le(n+1)^{O(\sqrt{2n})}=e^{O(\sqrt n\log n)}$; hence $T(\lambda)\le e^{O(\sqrt n\log n)}=
n^{O(\sqrt n)}$. The number of partitions is $p(n)=e^{\Theta(\sqrt n)}$ (Hardy--Ramanujan), so
$\sum_{\lambda\vdash n}T(\lambda)\le p(n)\cdot\max_\lambda T(\lambda)\le n^{O(\sqrt n)}$.
\end{proof}

\begin{remark}
The bound is sub-exponential but not polynomial: the partitions into distinct parts, of which there
are up to $r=\Theta(\sqrt n)$, force the $\sqrt n$ in the exponent. For families with a bounded
number of distinct part sizes, $P(\lambda)=n^{O(1)}$ and the cost is polynomial. The general
parameterized bound $O(2^{7\mu}n^9)$ of~\cite{HLV2025} is, for the cycle types with $\mu=\Theta(n)$
that dominate the Burnside sum, single-exponential in $n$; Theorem~\ref{thm:time} is the precise
sense in which the present method is feasible across the full partition set.
\end{remark}

\section{Computational results}\label{sec:results}

The recurrences of Section~\ref{sec:dp} and the Burnside sum~\eqref{eq:burnside} were implemented and
evaluated exactly in rational arithmetic. Table~\ref{tab:vals} lists both sequences; the entries for
$n\le15$ reproduce the previously published values, and $n=16$ through $20$ are new ($n=16$--$18$ have
been submitted to the OEIS; $n=19,20$ are reported here for the first time).

We stress that extending these sequences is not a matter of generating one more graph. The Burnside
sum~\eqref{eq:burnside} requires the exact value of $\fix(\lambda)$ for \emph{every} cycle type
$\lambda\vdash n$, and the number of cycle types is the partition number $p(n)$, which grows
sub-exponentially ($p(15)=176$, $p(20)=627$). The difficulty is concentrated in the cycle types with
many small orbits, where direct enumeration is hopeless and the equivariant recurrences of
Section~\ref{sec:dp} are essential; it is computing all of these $\fix(\lambda)$, not enumerating
graphs, that had kept the sequences at $n=15$.

\begin{table}[h]
\centering\small
\begin{tabular}{rrr}
\toprule
$n$ & A048193$(n)$ (all) & A048192$(n)$ (connected)\\
\midrule
8  & $2119$ & $1614$\\
9  & $14524$ & $11911$\\
10 & $126758$ & $109539$\\
11 & $1392387$ & $1247691$\\
12 & $19109099$ & $17566431$\\
13 & $326005775$ & $305310547$\\
14 & $6905776799$ & $6558690953$\\
15 & $181945055235$ & $174688164414$\\
\midrule
16 & $\mathbf{5985406996403}$ & $\mathbf{5796153514484}$\\
17 & $\mathbf{247178491630853}$ & $\mathbf{241003010628949}$\\
18 & $\mathbf{12895963060540295}$ & $\mathbf{12642592677074970}$\\
19 & $\mathbf{855912598965399807}$ & $\mathbf{842762851699294393}$\\
20 & $\mathbf{72786012927793961715}$ & $\mathbf{71916937400532750123}$\\
\bottomrule
\end{tabular}
\caption{Unlabeled chordal graph counts. Bold entries ($n\ge16$) are new.}
\label{tab:vals}
\end{table}

The computation was validated four independent ways.
\begin{enumerate}[topsep=2pt,itemsep=1pt]
\item \emph{Known data.} The Burnside sum reproduces all published values of both sequences for
$1\le n\le 15$ exactly.
\item \emph{Direct enumeration.} For every cycle type small enough to enumerate, $\fix(\lambda)$ and
$\cfix(\lambda)$ from the dynamic program agree exactly with a backtracking enumerator over orbit
edge-classes (which lays down all $\pi$-invariant edge-orbits and prunes on hereditary chordality).
This was checked across all orbit sizes $1\le s\le 16$, every single $c$-cycle, and many mixed and
composite types; a representative sample appears in Table~\ref{tab:valid}.
\item \emph{Labeled limit.} $\fix(1^n)$ reproduces the labeled chordal totals
\href{https://oeis.org/A058862}{A058862}$(n)$ for $1\le n\le 10$.
\item \emph{Euler transform.} Chordal graphs are closed under disjoint union, so the total count is
the Euler transform of the connected count; this identity holds for the computed values through
$n=20$, so the two independently assembled sequences are mutually consistent at every new term.
\end{enumerate}
For the two largest terms $n=19,20$, which lie beyond the reach of the backtracking oracle of~(2), the
validation rests on the method's exact reproduction of all known data~(1) and the labeled limit~(3),
together with the Euler cross-check~(4), under which the all-count and connected-count sequences
certify each other.

\begin{table}[h]
\centering\small
\begin{tabular}{lrrrc}
\toprule
$\lambda$ & $n$ & $\fix(\lambda)$ & $\cfix(\lambda)$ & brute\\
\midrule
$1^{3}$ & 3 & 8 & 4 & \checkmark \\
$2\,1^{2}$ & 4 & 15 & 7 & \checkmark \\
$2^{2}$ & 4 & 13 & 7 & \checkmark \\
$3\,1^{2}$ & 5 & 15 & 7 & \checkmark \\
$3\,2\,1$ & 6 & 26 & 11 & \checkmark \\
$3^{2}$ & 6 & 22 & 15 & \checkmark \\
$4\,1^{2}$ & 6 & 22 & 10 & \checkmark \\
$4\,2\,1$ & 7 & 69 & 37 & \checkmark \\
$4^{2}$ & 8 & 64 & 37 & \checkmark \\
$6\,3$ & 9 & 37 & 23 & \checkmark \\
$3^{3}$ & 9 & 897 & 745 & \checkmark \\
$8\,4$ & 12 & 92 & 52 & \checkmark \\
$9\,3$ & 12 & 34 & 22 & \checkmark \\
$10\,5$ & 15 & 113 & 95 & \checkmark \\
\bottomrule
\end{tabular}
\caption{Per-cycle-type values from the dynamic program, each agreeing exactly with direct
backtracking enumeration. The composite sizes $4,6,8,9,10$ and the cross-divisor types
($6\,3$, $8\,4$, $9\,3$, $10\,5$) exercise the divisor-bundle recursion.}
\label{tab:valid}
\end{table}

By Theorem~\ref{thm:time} the method extends to larger $n$ in sub-exponential time; the values above
were obtained on a single processor.

\paragraph{Data availability.} The two computed tables ($n=1,\dots,20$) have SHA-256 digests
\begin{center}\footnotesize\ttfamily
\begin{tabular}{ll}
A048192 (connected) & 2933d6dd7c2fb4133032158c4a5230514762f762975734e9df598f106ef0ac94\\
A048193 (all) & 95eda5ffd21902003b79624f16dd8fae021e4a00dfd3f78e7611e95824f95abb\\
\end{tabular}
\end{center}

\section{Conclusion}

We have given a concrete equivariant dynamic program for counting permutation-invariant chordal
graphs, built on the labeled evaporation program of~\cite{HLV2023}. Its correctness rests on a
divisor-bundle decomposition (Theorem~\ref{thm:cod}) that reduces each equivariant recurrence to its
labeled counterpart, and the full Burnside computation runs in sub-exponential time
(Theorem~\ref{thm:time}). The resulting enumeration extends the unlabeled chordal sequences A048192
and A048193 past their previous frontier. The orbit-bundle lifting strategy may be adaptable to other
graph classes that admit a compatible labeled decomposition recurrence; we leave such extensions to
future work.

\appendix

\section{The explicit equivariant recurrences}\label{app:rec}

We collect the eight recurrences over the state variables of Section~\ref{sec:state}. Vectors are
indexed by the sizes $S=(s_1<\cdots<s_r)$; $\binom{k}{m}:=\prod_i\binom{k_i}{m_i}$; $p=p(k)$ is the
largest index with $k_p>0$; and the marking coefficient is
\[
M(k,\kappa)=\binom{k_p-1}{\kappa_p-1}\prod_{i\ne p}\binom{k_i}{\kappa_i}.
\]
A bundle of period $d$ absorbs only sizes divisible by $d$; we write $\kappa\rhd d$ for ``$\kappa_i>0
\Rightarrow d\mid s_i$,'' and $\kappa/d$ for the vector with $s_i\mapsto s_i/d$. The boundary terms use
the block generating functions $\Gamma_d,\Gamma_d^{\mathrm{cov}}$ of~\eqref{eq:gf}; $[\,\Phi\,]_a$ is
the coefficient of $\prod_\tau v_\tau^{a_\tau}$, and $g_1^{\langle N/d\rangle}(t,a,\nu)$ is the
sub-world function with block-boundary $a$. Define the three peeled-component operators
(the $d=1$ summand $A$ and the $d>1$ summands $B_d$, for the three boundary conditions
$\Phi\in\{\mathrm{see},\mathrm{not\text{-}all},\mathrm{cover}\}$):
\begin{align*}
A^{\mathrm{see}}(\kappa)&=\sum_{0\ne x'\le x}\Big(\tbinom{x}{x'}-\tbinom{z}{x'}\Big)g_1(t,x',\kappa),\\
A^{\mathrm{not\text{-}all}}(\kappa)&=\sum_{0\ne x'\lneq x}\Big(\tbinom{x}{x'}-\tbinom{z}{x'}\Big)g_1(t,x',\kappa),\\
B_d^{\mathrm{see}}(\kappa)&=d^{|\kappa|-1}\sum_{a}\big[\Gamma_d(x)-\Gamma_d(z)\big]_a\,g_1^{\langle N/d\rangle}(t,a,\kappa/d),\\
B_d^{\mathrm{cov}}(\kappa)&=d^{|\kappa|-1}\sum_{a}\big[\Gamma_d^{\mathrm{cov}}(x)\big]_a\,g_1^{\langle N/d\rangle}(t,a,\kappa/d),\\
B_d^{\mathrm{not\text{-}all}}(\kappa)&=d^{|\kappa|-1}\sum_{a}\big[\Gamma_d(x)-\Gamma_d^{\mathrm{cov}}(x)-\Gamma_d(z)\big]_a\,g_1^{\langle N/d\rangle}(t,a,\kappa/d)\quad(0\text{ if }z=x).
\end{align*}

\paragraph{Base cases and invalid states.} Every function returns $0$ outside its domain and on
impossible states: $g_1=0$ if $t\le0$ or $k=0$; $g(0,x,k,z)=[\,k=0\,]$ and $g=0$ for $t<0$;
$\gtil(t,x,0,z)=1$; $\ghat(t,x,0,z)=1$ and $\ghat=0$ for $t<1$; $g_2=0$ if $t\le0$ or $k=0$;
$f=0$ for $t\le0$, $f(1,x,l,k)=[\,k=0\,]$, and $f=0$ for $k=0,\ t\ge2$; and $\ftil=\fhat=0$ for
$t\le1$ or $k=0$. Empty sums are $0$ and empty products $1$; $\binom{a}{b}=0$ for $b>a$ or $b<0$, so
a sub-multiset that does not fit contributes nothing, as does any state with a negative coordinate.
A clique-number bound $\omega$ is carried with the recursion, and $f=0$ once
$\sum_i s_i(x_i+l_i)>\omega$. The recursion is finite: $t$ strictly decreases through $f\to g\to\gtil$,
and the world order $N/d<N$ strictly decreases at each bundle step (Lemma~\ref{lem:sub}). Then:
\begin{align}
g_1(t,x,k)&=\sum_{0\ne\ell\le k}\tbinom{k}{\ell}\,f(t,x,\ell,k-\ell),\tag{R1}\\
f(t,x,l,k)&=\sum_{m\le k}\tbinom{k}{m}\,\ftil(t,x,l,m)\,g(t-2,x+l,k-m,x),\tag{R2}\\
g(t,x,k,z)&=\sum_{m\le k}\tbinom{k}{m}\,\gtil(t,x,m,z)\,g(t-1,x,k-m,z),\tag{R3}\\
\gtil(t,x,k,z)&=\sum_{\kappa_p\ge1}M(k,\kappa)\Big[A^{\mathrm{see}}(\kappa)+\!\!\sum_{d>1,\,\kappa\rhd d}\!\!B_d^{\mathrm{see}}(\kappa)\Big]\gtil(t,x,k-\kappa,z),\tag{R4}\\
\ghat(t,x,k,z)&=\sum_{\kappa_p\ge1}M(k,\kappa)\Big[A^{\mathrm{not\text{-}all}}(\kappa)+\!\!\sum_{d>1,\,\kappa\rhd d}\!\!B_d^{\mathrm{not\text{-}all}}(\kappa)\Big]\ghat(t,x,k-\kappa,z),\tag{R5}\\
g_2(t,x,k)&=\sum_{\kappa_p\ge1}M(k,\kappa)\Big[g_1(t,x,\kappa)\,R(k-\kappa)+\!\!\sum_{d>1,\,\kappa\rhd d}\!\!B_d^{\mathrm{cov}}(\kappa)\,R_0(k-\kappa)\Big],\tag{R6}\\
\ftil(t,x,l,k)&=\fhat(t,x,l,k)\notag\\
&\quad+\sum_{0\ne m\le k}\tbinom{k}{m}\Big[g_1(t{-}1,x{+}l,m)\,\fhat(t,x,l,k{-}m)+g_2(t{-}1,x{+}l,m)\,\ghat(t{-}1,x{+}l,k{-}m,x)\Big].\tag{R7}
\end{align}
where in (R6) $R(\cdot)=g_1+g_2$ and $R_0(m)=[\,m=0\,]+R(m)$. Finally the absorption recurrence: with
$\Gamma^{\mathrm{tou}}_d(\sigma)=\prod_i\big((1+v_{\tau_i})^{\gcd(d,s_i)}-1\big)^{\sigma_i}$ the
``touched-orbit'' generating function, the weight
$w(xp,lp)=\big(\binom{x}{xp}-[\,lp{=}0\,]\binom{z}{xp}\big)\binom{l}{lp}$, and the successor
$\mathrm{nx}_{lp}(k')=\ghat(t{-}1,x{+}lp,k',z)$ if $lp=l$ and $\fhat(t,x{+}lp,z,l{-}lp,k')$
otherwise,
\begin{equation}\tag{R8}
\fhat(t,x,z,l,k)=\!\!\sum_{\kappa_p\ge1}\!\!M(k,\kappa)\!\!\sum_{\substack{xp\le x,\ lp\le l\\ (xp,lp)\notin\{0,(x,l)\}}}\!\!\!\! w(xp,lp)\,\mathrm{nx}_{lp}(k{-}\kappa)\,\Theta_\kappa(xp{+}lp),
\end{equation}
where the touched-component value, the absorption analogue of the bracket in (R4), is
\[
\Theta_\kappa(\sigma)=g_1(t{-}1,\sigma,\kappa)+\!\!\sum_{d>1,\,\kappa\rhd d}\!\! d^{|\kappa|-1}\sum_a\big[\Gamma^{\mathrm{tou}}_d(\sigma)\big]_a\,g_1^{\langle N/d\rangle}(t{-}1,a,\kappa/d).
\]
Setting all sizes to $1$ ($r=1$, $N=1$, no $d>1$) recovers the labeled recurrences of
Theorem~\ref{thm:hlv}.

\section{A worked example: $\lambda=2^2$}\label{app:ex}

Let $\pi=(1\,2)(3\,4)$, so $N=2$ and the only sizes are $\{2\}$. We compute $\fix(2^2)$ from
Proposition~\ref{prop:fix} and Lemma~\ref{lem:bundle}, displaying the divisor-bundle terms.

A single component-bundle on $j$ of the $2$-cycles has, by Theorem~\ref{thm:cod},
\[
\mathrm{BUNDLE}(2^j)=\underbrace{\cfix(2^j)}_{d=1}+\underbrace{2^{\,j-1}\,\cfix^{\langle1\rangle}(1^j)}_{d=2},
\qquad \cfix^{\langle1\rangle}(1^j)=\text{\href{https://oeis.org/A007134}{A007134}}(j),
\]
the $d=2$ term being a swapped pair of labeled components. With $\cfix(2^1)=1$,
$\cfix(2^2)=7$ and $A007134(1)=A007134(2)=1$,
\[
\mathrm{BUNDLE}(2^1)=1+2^{0}\!\cdot\!1=2,\qquad
\mathrm{BUNDLE}(2^2)=7+2^{1}\!\cdot\!1=9.
\]
Proposition~\ref{prop:fix} with $k=(2)$ (marked size $2$) gives, over $\mu\in\{(1),(2)\}$,
\[
\fix(2^2)=\binom{1}{0}\mathrm{BUNDLE}(2^1)\,\fix(2^1)+\binom{1}{1}\mathrm{BUNDLE}(2^2)\,\fix(\varnothing).
\]
Here $\fix(2^1)=\mathrm{BUNDLE}(2^1)=2$ and $\fix(\varnothing)=1$, so
\[
\fix(2^2)=1\cdot2\cdot2+1\cdot9\cdot1=13,
\]
matching the direct enumeration in Table~\ref{tab:valid}. The contributions $2^{0}\cdot1$ and
$2^{1}\cdot1$ are exactly the period-$2$ bundles (an edge between the two cyclically-paired vertices,
in its $2^{\,j-1}$ alignments), which the $d=1$ count alone would miss.

\end{document}